\documentclass{amsart}
\usepackage{amssymb}
\usepackage[latin5]{inputenc}
\usepackage{amsfonts}

\newtheorem{theorem}{Theorem}
\theoremstyle{plain}

\newtheorem{corollary}{Corollary}

\newtheorem{definition}{Definition}

\newtheorem{lemma}{Lemma}

\newtheorem{proposition}{Proposition}

\numberwithin{equation}{section}

\newcommand{\Z}{\mathbb{Z}}

\begin{document}
\title[]{On Codes Over $\mathbb{Z}_{p^{s}}$ With Extended Lee Weight}
\author{Zeynep Ödemİş Özger}
\address{Zeynep Ödemİş Özger, Department of Mathematics, Fatih University, 34500, İstanbul, Turkey } \email{\tt zodemis@fatih.edu.tr}
\author{Bahattİn Yildiz}
\address{Bahattİn Yildiz, Department of Mathematics, Fatih University, 34500, İstanbul, Turkey } \email{\tt byildiz@fatih.edu.tr}
\author{Steven T. Dougherty}
\address{Steven T. Dougherty, Department of Mathematics,  University of Scranton, Scranton, PA 18510, USA } \email{\tt prof.steven.dougherty@gmail.com}


\begin{abstract}
We consider codes over $\Z_{p^s}$ with the extended Lee weight. We find Singleton bounds with respect to this weight and define  MLDS and MLDR codes accordingly. We also consider the kernels of these codes and the notion of independence of vectors in this space. We investigate the linearity and duality of the Gray images of codes over $\Z_{p^s}$.
\end{abstract}

\maketitle

\section{Introduction}

In the early history of coding theory, codes over finite fields were
predominantly studied. The most common weight used for such codes was the
Hamming weight, which is defined to be  the number of nonzero coordinates.  We will denote by the hamming weight by  $w_{H}$. Many encoding and decoding schemes
as well as error correction algorithms are based on the Hamming distance.

Codes over rings have been considered since the early seventies, however it was not until the beginning of the
nineties that they became a widely popular research field in coding theory. In 1994, Hammons
et al.(\cite{hammons}) solved a long standing problem in non-linear binary
codes by constructing the Kerdock and Preparata codes as the Gray images of
linear codes over $\mathbb{Z}_{4}$. This work started an intense activity on
codes over rings. The rich algebraic structure that rings bring together
with some better than optimal nonlinear codes obtained from linear codes
over rings have increased the popularity of this topic. What started with
the ring $\mathbb{Z}_{4}$, later was extended to rings such as $\mathbb{Z}%
_{2^{k}}$, $\mathbb{Z}_{p^{k}}$, Galois rings, $\mathbb{F}_{q}+u\mathbb{F}%
_{q}$, and various other rings.

For codes over rings, weights other than the Hamming weight were considered. For example, in \cite{hammons}, the authors used the Lee weight
on $%
\mathbb{Z}
_{4}$, which we will denote by $w_{L}$ and was defined as
\begin{equation*}
w_{L}(x):=\left\{
\begin{array}{ll}
0 & \text{if }x=0\text{,} \\
2 & \text{if }x=2\text{,} \\
1 & \text{otherwise.}%
\end{array}%
\right.
\end{equation*}

The Gray map
\begin{equation*}
\phi _{L}:%
\mathbb{Z}
_{4}\rightarrow
\mathbb{Z}
_{2}^{2}\text{,}
\end{equation*}%
with
\begin{equation*}
\phi _{L}(0)=(00)\text{, }\phi _{L}(1)=(01)\text{, }\phi _{L}(2)=(11)\text{,
}\phi _{L}(3)=(10)\text{,}
\end{equation*}%
turns out to be a non-linear isometry from $(\mathbb{Z}_{4}^{n},\text{%
Lee distance})$ to $(\mathbb{F}_{2}^{2n},\text{Hamming distance})$. This
means that if $C$ is a linear code over $\mathbb{Z}_{4}$ of length $n$, size
$M$ and minimum Lee distance $d$, then $\phi _{L}(C)$ is a possibly
non-linear binary code with parameters $[n,M,d]$.

When extending the Lee distance from $\mathbb{Z}_4$ to the more general ring extensions, the homogeneous weight was mostly used. The homogeneous weight has a lot of advantages, which made them useful in constructing codes over rings. It is related to exponential sums (see   \cite{consheise} and \cite{voloch} for example), making it easier to find bounds by using some number theoretic arguments such as the Weil bound. The homogeneous weight also gives rise to codes with high divisibility properties.

Another extension of the Lee weight is also possible and has been used by different researchers. For example the weight $w_L$ on $\mathbb{Z}_{2^s}$, defined by
\begin{equation*}
w _{L}(x)=\left\{
\begin{array}{ll}
x & \text{if }x\leq 2^{s-1}\text{,} \\
2^{s}-x & \text{if }x>2^{s-1}\text{.}%
\end{array}%
\right.
\end{equation*} was used partly in \cite{carlet}, \cite{dougherty} and \cite{aipyildizoo}. A simple Gray map for this weight maps codes over $\mathbb{Z}_{2^s}$ to (mostly) nonlinear binary codes.

This extension was generalized to $\mathbb{Z}_m$ as the Lee weight by letting $w_L(x) = \min\{x, m-x\}$ in some works, however no Gray map has been offered for such a weight.

In this work, we generalize the Lee weight on $\mathbb{Z}_{2^s}$ given above to the rings $\mathbb{Z}_{p^s}$ and the Galois rings $GR(p^s,m)$, together with a simple description of a Gray map projecting codes over $\Z_{p^s}$ to codes over the finite prime field $\mathbb{F}_p = \mathbb{Z}_p$. We study codes over $\mathbb{Z}_{p^s}$ together with this Lee weight from many angles such as Singleton bounds, independence, kernels and duality.

The rest of the paper is organized as follows: In Section 2, we recall the extended Lee weight, the Gray map and some properties for codes over $\mathbb{Z}_{p^s}$ from \cite{twmsyildizoo}. In Section 3 some
bounds on codes over $\mathbb{Z}_{p^{s}}$\ concerning both length and size
of the codes are given and MLDS and MLDR codes are defined accordingly. In
Section 4 the notions of kernel and independence are investigated. In
Section 5 some results about self-duality and self-orthogonality are found.

\section{The Extended Lee Weight and Its Gray Map}

We recall that a new weight on $\mathbb{Z}_{p^{s}}$, a generalization of $%
w_{L} $, was defined in \cite{twmsyildizoo} as follows:
\begin{equation*}
w_{L}(x):=\left\{
\begin{array}{lll}
x & \text{if }x\leq p^{s-1}\text{,} &  \\
p^{s-1} & \text{if }p^{s-1}\leq x\leq p^{s}-p^{s-1}\text{,} &  \\
p^{s}-x & \text{if }p^{s}-p^{s-1}<x\leq p^{s}-1\text{,} &
\end{array}%
\right.
\end{equation*}%
where $p$ is prime. Note that for $p=2$ and $s=2$ this reduces to the Lee
weight for $\mathbb{Z}_{4}$ and for $p=2$ and any $s$, this is the weight
that was used briefly by Carlet in \cite{carlet} and by Dougherty and Fern\'{a}ndez-C\'{o}rdoba in \cite%
{dougherty}. We can define a
Gray map from $\mathbb{Z}_{p^{s}}$ to $\mathbb{Z}_{p}^{p^{s-1}}$ just as was
done for the homogeneous weight as follows:

\begin{equation*}
\begin{array}{lll}
0 & \rightarrow & (000\cdot \cdot \cdot 000)\text{,} \\
1 & \rightarrow & (100\cdot \cdot \cdot 000)\text{,} \\
2 & \rightarrow & (110\cdot \cdot \cdot 000)\text{,} \\
& \cdot &  \\
& \cdot &  \\
p^{s-1} & \rightarrow & (111\cdot \cdot \cdot 111)\text{,} \\
p^{s-1}+1 & \rightarrow & (211\cdot \cdot \cdot 111)\text{,} \\
p^{s-1}+2 & \rightarrow & (221\cdot \cdot \cdot 111)\text{,} \\
& \cdot &  \\
& \cdot &  \\
p^{s-1}+p^{s-1}-1 & \rightarrow & (222\cdot \cdot \cdot 221)\text{,} \\
2p^{s-1} & \rightarrow & (222\cdot \cdot \cdot 222)\text{,} \\
2p^{s-1}+1 & \rightarrow & (322\cdot \cdot \cdot 222)\text{,} \\
& \cdot &  \\
& \cdot &  \\
2p^{s-1}+p^{s-1}-1 & \rightarrow & (333\cdot \cdot \cdot 332)\text{,} \\
3p^{s-1} & \rightarrow & (333\cdot \cdot \cdot 333)\text{,} \\
& \cdot &  \\
& \cdot &  \\
(p-1)p^{s-1} & \rightarrow & ((p-1)\cdot \cdot \cdot (p-1))\text{,} \\
(p-1)p^{s-1}+1 & \rightarrow & (0(p-1)\cdot \cdot \cdot (p-1))\text{,} \\
& \cdot &  \\
& \cdot &  \\
p^{s}-2 & \rightarrow & (000\cdot \cdot \cdot 0(p-1)(p-1))\text{,} \\
p^{s}-1 & \rightarrow & (000\cdot \cdot \cdot 00(p-1))\text{.}%
\end{array}%
\end{equation*}%
We simply put a $1$ in the first $x$ coordinates and a $0$ in the other
coordinates for all $x\leq p^{s-1}$. If $x>p^{s-1}$ then the Gray map takes $%
x$ to $\overline{q}+\phi _{L}(r)$, where $\phi _{L}$ is the Gray map for $%
w_{L}$, $\overline{q}=(qqq\cdot \cdot \cdot qqq)$ and $q$ and $r$ are such
that%
\begin{equation*}
x=qp^{s-1}+r\text{,}
\end{equation*}%
which can be found by division algorithm. Here, $0\leq x\leq p^{s}-1$, $%
0\leq q\leq p-1$, $0\leq r\leq p^{s-1}-1$. Here by putting $p=2$, we get the
same Gray map given in \cite{aipyildizoo} and \cite{dougherty}, which is%
\begin{equation*}
\begin{array}{lll}
0 & \rightarrow & (000\cdot \cdot \cdot 000) \\
1 & \rightarrow & (100\cdot \cdot \cdot 000) \\
2 & \rightarrow & (110\cdot \cdot \cdot 000) \\
& \cdot &  \\
& \cdot &  \\
2^{s-1} & \rightarrow & (111\cdot \cdot \cdot 111) \\
2^{s-1}+1 & \rightarrow & (011\cdot \cdot \cdot 111) \\
2^{s-1}+2 & \rightarrow & (001\cdot \cdot \cdot 111) \\
& \cdot &  \\
& \cdot &  \\
2^{s}-2 & \rightarrow & (000\cdot \cdot \cdot 011) \\
2^{s}-1 & \rightarrow & (000\cdot \cdot \cdot 001)\text{.}%
\end{array}%
\end{equation*}

As an example, when $p=3$, $s=2$ we get the extended Lee weight on $\Z_9$, which is a non-homogenous weight and is defined as
\begin{equation*}
w_{L}(x):=\left\{
\begin{array}{lll}
x & \text{if }x\leq 3\text{,} &  \\
3 & \text{if }3\leq x\leq 6\text{,} &  \\
9-x & \text{if }6 <x\leq 8\text{,} &
\end{array}%
\right.
\end{equation*}%
The Gray map takes $\Z_9$ to $\Z_3^3$ as follows:
$$\begin{array}{lll}
0 & \rightarrow & (000) \\
1 & \rightarrow & (100) \\
2 & \rightarrow & (110) \\
3 & \rightarrow & (111) \\
4 & \rightarrow & (211) \\
5 & \rightarrow & (221) \\
6 & \rightarrow & (222) \\
7 & \rightarrow & (022) \\
8 & \rightarrow & (002)\text{.}%
\end{array}$$

We define the Lee distance on $\mathbb{Z}_{p^{s}}$ as
\begin{equation}
d_{L}(x,y):=w_{L}(x-y), \:\:\:\:\:\:x,y\in \mathbb{Z}_{p^{s}}.
\end{equation}%
Note that this is a metric on $\mathbb{Z}_{p^{s}}$ and by extending $w_{L}$
and $d_{L}$ linearly to $(\mathbb{Z}_{p^{s}})^{n}$ in an obvious way, we get
a weight and a metric on $(\mathbb{Z}_{p^{s}})^{n}$.

\begin{theorem}
The map $\phi _{L}:(%
\mathbb{Z}
_{p^{s}},d_{L})\longrightarrow (\mathbb{F}_{p}^{p^{s-1}},d_{H})$ is a
distance preserving (not necessarily linear) map, where $d_{L}$ and $d_{H}$
denote the Lee and the Hamming distances respectively.
\end{theorem}

The proof of this theorem can be found in \cite{twmsyildizoo} with the
following corollary:

\begin{corollary}
If $C$ is a linear code over $\mathbb{Z}_{p^{s}}$ of length $n$, size $M$
and minimum \textit{Lee} distance $d$, then $\phi _{L}(C)$ is a (possibly
non-linear) code over $\mathbb{F}_{p}$ of length $np^{s-1}$, size $M$ and
minimum \textit{Hamming} distance $d$.
\end{corollary}

A Gray map from $GR(p^{s},m)$ to $\mathbb{F}_{p}^{p^{s-1}m}$ can also be
defined by extending this map (see \cite{twmsyildizoo}, Section 3), which
means that most of the work done in this paper is applicable to Galois rings.

\section{Singleton Bounds For Codes Over $%
\mathbb{Z}
_{p^{s}}$}

A Singleton bound for codes over a finite quasi-Frobenius ring is already
given in \cite{shiromoto2} as an MDS bound. Since this result is given for
any weight function, it can be specified for the extended Lee weight.

\begin{definition}[Complete weight]
\label{compwfuncdefn}\cite{shiromoto2} Let $R$ be a finite commutative
quasi-Frobenius ring, and let $V:=R^{n}$ be a free module of rank $n$
consisting of all $n$-tuples of elements of $R$. For every $x=(x_{1},\cdot
\cdot \cdot ,x_{n})\in V$ and $r\in R$, the complete weight of $x$ is
defined by%
\begin{equation}
n_{r}(x):=\left\vert \left\{ i\left\vert x_{i}=r\right. \right\} \right\vert
\text{.}  \label{comwfunc}
\end{equation}
\end{definition}

\begin{definition}[General weight function]
\label{genwfuncdefn}\cite{shiromoto2} Let $a_{r}$,$(0\neq )r\in R$, be
positive real numbers, and set $a_{0}=0$. Then%
\begin{equation}
w(x):=\sum\limits_{r\in R}a_{r}n_{r}(x)  \label{genwfunc}
\end{equation}%
is called a general weight function.
\end{definition}

Note that when $a_{r}=1$, $r\in R - \{0 \}$, $w(x)$ gives the Hamming weight
of $x$.

The following theorem gives a Singleton bound for any finite quasi-Frobenius
ring and any weight function.

\begin{theorem}
\label{mdsthmshiro}\cite{shiromoto2} Let $C$ be a code of length $n$ over a
finite commutative $QF$ ring $R$. Let $w(x)$ be a general weight function on
$C$, as in (\ref{genwfunc}), and with maximum $a_{r}-$value $A$. Suppose the
minimum weight of $w(x)$ on $C$ is $d$. Then%
\begin{equation}
\left\lfloor \frac{d-1}{A}\right\rfloor \leq n-\log _{\left\vert
R\right\vert }\left\vert C\right\vert \text{,}  \label{mdsgeneral}
\end{equation}%
where $\left\lfloor b\right\rfloor $ is the integer part of $b$.
\end{theorem}

Since $%
\mathbb{Z}
_{p^{s}}$ is a finite commutative Frobenius ring by letting $w(x)=w_{L}(x)$,
we have $p^{s-1}$ as the maximum $a_{r}-$value. Applying these informations
to Theorem we get the following:

\begin{theorem}
\label{mdsbound}Let $C$ be a code of length $n$ over $%
\mathbb{Z}
_{p^{s}}$ with minimum distance $d$. Then%
\begin{equation}
\left\lfloor \frac{d-1}{p^{s-1}}\right\rfloor \leq n-\log _{p^{s}}\left\vert
C\right\vert \text{.}  \label{mdslee}
\end{equation}
\end{theorem}

Codes meeting this bound are called MLDS (Maximum Lee Distance Separable)
codes. In \cite{shiromoto}, another bound was found over $%
\mathbb{Z}
_{l}$ with a different generalization of the Lee weight. Now we will
find a similar result for codes over $%
\mathbb{Z}
_{p^{s}}$\ with $w_{L}(x)$ by the same method used.

\begin{definition}[Rank, Free-rank]
\label{rankfreerank}Let $C$ be any finitely generated submodule of $%
\mathbb{Z}
_{p^{s}}^{n}$, that is isomorphic to%
\begin{equation}
\mathbb{Z}
_{p^{s}}/p^{a_{1}}%
\mathbb{Z}
_{p^{s}}\oplus
\mathbb{Z}
_{p^{s}}/p^{a_{2}}%
\mathbb{Z}
_{p^{s}}\oplus \cdot \cdot \cdot \oplus
\mathbb{Z}
_{p^{s}}/p^{a_{n-1}}%
\mathbb{Z}
_{p^{s}}\text{,}  \label{submoduleisom}
\end{equation}%
where $a_{i}$ are positive integers with $p^{a_{1}}|p^{a_{2}}|\cdot \cdot
\cdot |p^{a_{n-1}}|p^{s}$. Then%
\begin{equation}
rank(C):=\left\vert \left\{ i\left\vert a_{i}\neq 0\right. \right\}
\right\vert \text{,}  \label{rank}
\end{equation}%
is called the rank of $C$ and%
\begin{equation}
free~rank(C):=\left\vert \left\{ i\left\vert a_{i}=s\right. \right\} \right\vert
\label{frank}
\end{equation}%
is called the free rank of $C$.
\end{definition}


Any code over $\Z_{p^s}$ has a generator matrix of the form:
\begin{equation}
G=\left[
\begin{array}{ccccccc}
I_{\delta _{0}} & A_{0,1} & A_{0,2} & A_{0,3} & \cdot \cdot \cdot & \cdot
\cdot \cdot & A_{0,s} \\
0 & pI_{\delta _{1}} & pA_{1,2} & pA_{1,3} & \cdot \cdot \cdot & \cdot \cdot
\cdot & pA_{1,s} \\
0 & 0 & p^{2}I_{\delta _{2}} & p^{2}A_{2,3} & \cdot \cdot \cdot & \cdot
\cdot \cdot & p^{2}A_{2,s} \\
\cdot \cdot \cdot & \cdot \cdot \cdot & 0 & \cdot \cdot \cdot & \cdot \cdot
\cdot & \cdot \cdot \cdot & \cdot \cdot \cdot \\
\cdot \cdot \cdot & \cdot \cdot \cdot & \cdot \cdot \cdot & \cdot \cdot \cdot
& \cdot \cdot \cdot & \cdot \cdot \cdot & \cdot \cdot \cdot \\
0 & 0 & 0 & \cdot \cdot \cdot & p^{s-2}I_{\delta _{s-2}} & p^{s-2}A_{s-2,s-1}
& p^{s-2}A_{s-2,s} \\
0 & 0 & 0 & \cdot \cdot \cdot & 0 & p^{s-1}I_{\delta _{s-1}} &
p^{s-1}A_{s-1,s}%
\end{array}%
\right] \text{.}  \label{genmatc}
\end{equation}%

Then a code $C$ over $%
\mathbb{Z}
_{p^{s}}^{n}$ is of type $(p^{s})^{\delta _{0}}(p^{s-1})^{\delta _{1}}\cdot
\cdot \cdot (p)^{\delta _{s-1}}$, and %
\begin{eqnarray*}
rank(C) &=&\delta _{0}+\delta _{1}+\cdot \cdot \cdot +\delta _{s-1}\text{,}
\\
free~rank(C) &=&\delta _{0}
\end{eqnarray*}%

\ Let $C^{\bot }$, namely\ the dual of $C$, be defined as%
\begin{equation*}
C^{\bot }=\left\{ v\in
\mathbb{Z}
_{p^{s}}^{n}|\left\langle v,w\right\rangle =0\text{ for all }w\in C\right\}
\text{,}
\end{equation*}%
where $\left\langle v,w\right\rangle =\sum v_{i}w_{i}$ (mod $p^{s}$).
The code  $C^{\bot }$ is
isomorphic to%
\begin{equation*}
\mathbb{Z}
_{p^{s}}/p^{s-a_{1}}%
\mathbb{Z}
_{p^{s}}\oplus
\mathbb{Z}
_{p^{s}}/p^{s-a_{2}}%
\mathbb{Z}
_{p^{s}}\oplus \cdot \cdot \cdot \oplus
\mathbb{Z}
_{p^{s}}/p^{s-a_{n-1}}%
\mathbb{Z}
_{p^{s}}\text{.}
\end{equation*}%
From
\cite{shiromoto}, \cite{modipzmpark}, \cite{indepcodoverringsdoug}, \cite%
{dougshiromdr}, \cite{dougherty}, and the definitions above, the
relationship between the rank of a code and its dual's free rank can be
given as follows:%
\begin{equation}
rank(C)+free~rank(C^{\bot })=n  \label{ranknull}
\end{equation}%
For a submodule $D\subseteq V:=(%
\mathbb{Z}
_{p^{s}})^{n}$ and a subset $M\subseteq N:=\left\{ 1,2,\cdot \cdot \cdot
,n\right\} $, we define%
\begin{equation}
\begin{array}{l}
D(M):=\left\{ x\in D\left\vert \text{supp}(x)\subseteq M\right. \right\}
\text{,} \\
D^{\ast }:=Hom_{%
\mathbb{Z}
_{p^{s}}}(D,%
\mathbb{Z}
_{p^{s}})\text{,}%
\end{array}
\label{dhom}
\end{equation}%
where%
\begin{equation}
\text{supp}(x):=\left\{ i\in N\left\vert x_{i}\neq 0\right. \right\} \text{.}
\label{support}
\end{equation}%
From the fundamental theorem of finitely generated abelian groups, we have $%
D^{\ast }\cong D$. Shiromoto also gave the following basic exact sequence:

\begin{lemma}
\label{basicexactlemma}\cite{shiromoto}Let $C$ be a code of length $n$ over $%
\mathbb{Z}
_{l}$ and $M\subseteq N$. Then there is an exact sequence as $%
\mathbb{Z}
_{l}$-modules
\begin{subequations}
\begin{equation}
0\rightarrow C^{\bot }(m)\overset{inc}{\rightarrow }V(M)\overset{f}{%
\rightarrow }C^{\ast }\overset{res}{\rightarrow }C(N-M)^{\ast }\rightarrow 0%
\text{,}  \label{basicexact}
\end{equation}%
where the maps $inc$, $res$ denote the inclusion map, the restriction map,
respectively, and $f$ is a $%
\mathbb{Z}
_{l}$-homomorphism such that
\end{subequations}
\begin{equation}
\begin{array}{ll}
f: & V\rightarrow D^{\ast } \\
& y\rightarrow (\hat{y}:x\rightarrow \left\langle x,y\right\rangle \text{.}%
\end{array}
\label{fhom}
\end{equation}
\end{lemma}

We can adjust Lemma \ref{basicexactlemma} to our case:

\begin{lemma}
\label{mybasicexactlemma}Let $C$ be a code of length $n$ over $%
\mathbb{Z}
_{p^{s}}$ and $M\subseteq N$. Then there is an exact sequence as $%
\mathbb{Z}
_{p^{s}}$-modules
\begin{subequations}
\begin{equation}
0\rightarrow C^{\bot }(m)\overset{inc}{\rightarrow }V(M)\overset{f}{%
\rightarrow }C^{\ast }\overset{res}{\rightarrow }C(N-M)^{\ast }\rightarrow 0%
\text{,}  \label{mybasicexact}
\end{equation}%
where the maps $inc$, $res$ denote the inclusion map, the restriction map,
respectively, and $f$ is a $%
\mathbb{Z}
_{p^{s}}$-homomorphism such that
\end{subequations}
\begin{equation}
\begin{array}{ll}
f: & V\rightarrow D^{\ast } \\
& y\rightarrow (\hat{y}:x\rightarrow \left\langle x,y\right\rangle )\text{.}%
\end{array}
\label{myfhom}
\end{equation}
\end{lemma}

Note that for any $x\in V$, if supp$(x)\subseteq M\subseteq N$, then for any
general weight function we have $wt(x)\leq a_{r}\left\vert M\right\vert $.
In our case:%
\begin{equation}
w_{L}(x)\leq p^{s-1}\left\vert M\right\vert \text{.}  \label{supportineq}
\end{equation}%
So we have the following lemma for $w_{L}(x)$:

\begin{lemma}
\label{singboundlemma}Let $C$ be a code of length $n$ over $\mathbb{Z}%
_{p^{s}}$, then $C(M)^{\ast }=0$ for any subset $M\subseteq N$ such that $%
\left\vert M\right\vert <d/p^{s-1}$, where $d$ is the minimum Lee weight.

\begin{proof}
For any $\overset{\_}{c}\neq 0\in C$%
\begin{equation}
\left\vert \text{supp}(\overset{\_}{c})\right\vert p^{s-1}\geq w_{L}(\overset%
{\_}{c})\geq d\text{.}  \label{sup1}
\end{equation}%
If $\left\vert M\right\vert <d/p^{s-1}$, then%
\begin{equation}
d>\left\vert M\right\vert p^{s-1}\text{,}  \label{sup2}
\end{equation}%
which means%
\begin{equation}
\left\vert \text{supp}(\overset{\_}{c})\right\vert p^{s-1}\geq d>\left\vert
M\right\vert p^{s-1}  \label{sup3}
\end{equation}%
by (\ref{sup1}) and (\ref{sup2}). But this means $\left\vert \text{supp}(%
\overset{\_}{c})\right\vert >\left\vert M\right\vert $, i.e. supp$(\overset{%
\_}{c})\nsubseteq M$. So $C\cap V(M)=\left\{ 0\right\} $ and $C(M)^{\ast
}=Hom_{\mathbb{Z}_{p^{s}}}(C\cap V(M),\mathbb{Z}_{p^{s}})=0$.
\end{proof}
\end{lemma}

By Lemma \ref{singboundlemma}, we have the following bound:

\begin{theorem}
\label{singboundthm}Let $C$ be a code of length $n$ over $\mathbb{Z}_{p^{s}}$
with the minimum Lee weight $d$. Then%
\begin{equation}
\left\lfloor \frac{d-1}{p^{s-1}}\right\rfloor \leq n-rank(C)\text{.}
\label{mdrlee}
\end{equation}
\end{theorem}

\begin{proof}
We will follow the steps of Shiromoto in \cite{shiromoto}. In the exact
sequence of Lemma \ref{mybasicexactlemma}, replace $C$ with $C^{\bot }$.
Then the exact sequence transforms into the following one:%
\begin{equation}
0\rightarrow C(M)\overset{inc}{\rightarrow }V(M)\overset{f}{\rightarrow }%
(C^{\bot })^{\ast }\overset{res}{\rightarrow }C^{\bot }(N-M)^{\ast
}\rightarrow 0\text{.}  \label{dualbasicexact}
\end{equation}%
Apply $\ast =Hom_{\mathbb{Z}_{p^{s}}}(\cdot ,\mathbb{Z}_{p^{s}})$ and take
an arbitrary subset $M\subseteq N$ such that%
\begin{equation*}
\left\vert M\right\vert =\left\lfloor \frac{d-1}{p^{s-1}}\right\rfloor \text{%
.}
\end{equation*}%
Since $C(M)^{\ast }=0$ by Lemma \ref{singboundlemma} and $V(M)^{\ast }\cong
V(M)$, the exact sequence (\ref{dualbasicexact}) leads us to the following
short exact sequence:%
\begin{equation}
0\rightarrow C^{\bot }(N-M)\rightarrow C^{\bot }\rightarrow V(M)\rightarrow 0%
\text{.}  \label{shortexact}
\end{equation}%
$V(M)\cong (\mathbb{Z}_{p^{s}})^{\left\vert M\right\vert }$ is a projective
module. Hence (\ref{shortexact}) is a split, that is,%
\begin{equation*}
C^{\bot }\cong C^{\bot }(N-M)\oplus V(M)\text{.}
\end{equation*}%
Therefore%
\begin{equation*}
free~rank(C^{\bot })\geq free~rank(V(M))=\left\vert M\right\vert =\left\lfloor
\frac{d-1}{p^{s-1}}\right\rfloor \text{.}
\end{equation*}%
From (\ref{ranknull}) we have%
\begin{equation*}
n-rank(C)\geq \left\lfloor \frac{d-1}{p^{s-1}}\right\rfloor \text{.}
\end{equation*}
\end{proof}

Codes meeting the bound above are called MLDR (Maximum Lee distance with
respect to Rank) codes.

\section{Kernel and Independence of $\protect\phi_{L} (C)$}

For finite fields and vector spaces the notions of kernel and independence
are strongly related. In this section, we investigate the same notions for
Gray images of linear codes over $\mathbb{Z}_{p^{s}}$, which can be seen as $%
\mathbb{Z}_{p^{s}}$-submodules of $\mathbb{Z}_{p^{s}}^{n}$. The kernel of a
code $C$, denoted by $K(C)$, is defined as the set%
\begin{equation*}
K(C)=\left\{ v\left\vert v\in C,v+C=C\right. \right\} \text{.}
\end{equation*}%
Since $\phi_{L} (C)$ is a code (not necessarily linear), we can define%
\begin{equation*}
K(\phi_{L} (C))=\left\{ \phi_{L} (v)\left\vert v\in C,\phi_{L} (v)+\phi_{L}
(C)=\phi_{L} (C)\right. \right\} \text{.}
\end{equation*}%
In \cite{dougherty}, authors gave some results about $K(\phi_{L} (C))$, $%
\phi_{L} $-independence and modular independence over $\mathbb{Z}_{2^{s}}$.
We have similar results for $\mathbb{Z}_{p^{s}}$.

First we define modular independence.  We say that vectors ${v_1}, {v_2},{v_t}$ are modular independent over
$\Z_{p^s}$ if $\sum \alpha_i {v_i} = {\bf 0}$ then $\alpha_i \in \langle p \rangle$ for all $i$.

\begin{lemma}
\label{independencelemma}Let $G$ be the generating matrix of a linear code
of type $(p^{s})^{\delta _{0}}(p^{s-1})^{\delta _{1}}\cdot \cdot \cdot
(p)^{\delta _{s-1}}$ over $\mathbb{Z}_{p^{s}}$ in standard form. Let $%
v_{i,1},v_{i,2},\cdot \cdot \cdot ,v_{i,\delta _{i}}$ be the vectors of
order $p^{s-i}$. Then the vectors in the set $\left\{ \alpha v_{i,j}|1\leq
\alpha \leq p^{s-i-1}\right\} $ are $\phi_{L} $-independent in $\mathbb{F}%
_{p}^{p^{s-1}n}$.
\end{lemma}

\begin{proof}
Let $G$ be the generator matrix of the code as given in \ref{genmatc}.
The Gray images of $1,2,\cdot \cdot \cdot ,p^{s-1}$ form an upper triangular
matrix and so the Gray image of the vectors in the first $\delta _{0}$
coordinates are linearly independent. All initial nonzero coordinates of
submatrices $p^{i}I_{\delta _{i}}$ form an uppertriangular matrix and their
entries are all less than or equal to $p^{s-1}$. Therefore the other cases
of the form $p^{i}I_{\delta _{i}}$ form submatrix of the above mentioned
upper triangular matrix. Hence they are also linearly independent.
\end{proof}

\begin{theorem}
\label{independencetheorem}Let $v_{1},v_{2},\cdot \cdot \cdot ,v_{k}$ be
modular independent vectors in $\mathbb{Z}_{p^{s}}^{n}$. Then there exist
modular independent vectors $w_{1},w_{2},\cdot \cdot \cdot ,w_{k}$ which are
$\phi_{L} $-independent in $\mathbb{F}_{p}^{p^{s-1}n}$ such that $%
\left\langle v_{1},v_{2},\cdot \cdot \cdot ,v_{k}\right\rangle =\left\langle
w_{1},w_{2},\cdot \cdot \cdot ,w_{k}\right\rangle $.
\end{theorem}

\begin{proof}
Any set of modular independent vectors over $\mathbb{Z}_{p^{s}}$ are
permutationally equivalent to a set of vectors that form a generator matrix
in standard form as shown in \cite{modipzmpark}. Therefore by Lemma \ref%
{independencelemma} these vectors are $\phi_{L} $-independent.
\end{proof}

The following proposition gives a restriction to the order of elements whose
Gray images belong to $K(\phi_{L} (C))$.

\begin{proposition}
\label{psquareprop}Let $C$ be a linear code over $\mathbb{Z}_{p^{s}}$. If $%
v\in C$ has order greater than $p^{2}$ then $K(\phi_{L} (C))$ does not
contain $\phi_{L} (v)$.
\end{proposition}

\begin{proof}
Since $ord(v)>p^{2}$, $v$ has a number $i$ as its coordinate with $%
ord(i)>p^{2}$. We have the following three cases for $i\in \mathbb{Z}%
_{p^{s}} $ with $ord(i)>p^{2}$:

\begin{description}
\item[(i)] If $0<i<p^{s-1}$ then $ord(i)=p^{k}$, $k>2$, since $%
ord(i)|\left\vert \mathbb{Z}_{p^{s}}\right\vert =p^{s}$. That means $%
i=p^{s-k}u_{i}$, where $(u_{i},p^{s})=1$, i.e., $(u_{i},p)=1$. Since $%
0<i<p^{s-1}$
\begin{equation*}
\phi_{L} (i)=\overline{1}_{i}\overline{0}_{P^{s-1}-i}\text{,}
\end{equation*}%
and since $s-k\leq s-2$, we have $pi=p^{s-k+1}u_{i}<p^{s}$. We know that $%
i\neq p^{s-1}u_{j}$ or $i\neq p^{s-2}u_{j}$ for any $u_{j}$ such that $%
(u_{j},p)=1$. So by using division algorithm we can write%
\begin{equation*}
\begin{array}{ll}
i=qp^{s-2}+r^{\prime }\text{,} & 0<r^{\prime }<p^{s-2}\text{,} \\
pi=qp^{s-1}+r\text{,} & 0<r=pr^{\prime }<p^{s-1}\text{.}%
\end{array}%
\end{equation*}%
Without loss of generality assume that $i>r$. Then,%
\begin{eqnarray*}
\phi_{L} (i)+\phi_{L} (pi) &=&\overline{1}_{i}\overline{0}_{P^{s-1}-i}+%
\overline{q}_{p^{s-1}}+\overline{1}_{r}\overline{0}_{P^{s-1}-r} \\
&=&\overline{q+2}_{r}\overline{q+1}_{i-r}\overline{q}_{p^{s-1}-i}\notin
\phi_{L} (C)\text{,}
\end{eqnarray*}%
since $r\neq 0$, $r-i\neq 0$ and $p^{s-1}-i\neq 0$. Now assume $i=r$. Then,%
\begin{eqnarray*}
\phi_{L} (i)+\phi_{L} (pi) &=&\overline{1}_{i}\overline{0}_{P^{s-1}-i}+%
\overline{q}_{p^{s-1}}+\overline{1}_{i}\overline{0}_{P^{s-1}-i} \\
&=&\overline{q+2}_{i}\overline{q}_{p^{s-1}-i}\notin \phi_{L} (C)\text{,}
\end{eqnarray*}%
since $i\neq 0$ and $p^{s-1}-i\neq 0$.

\item[(ii)] If $p^{s-1}<i<p^{s}-p^{s-1}$ then $mp^{s-1}<i<(m+1)p^{s-1}$,
where $m\in \left\{ 1,2,3,\cdot \cdot \cdot ,p-2\right\} $. Since $%
ord(i)>p^{2}$, $i\neq p^{s-1}u_{j}$ or $i\neq p^{s-2}u_{j}$ for any $%
u_{j}\in \left\{ 1,2,3,\cdot \cdot \cdot ,p-2,p-1\right\} $. Let
\begin{equation*}
\begin{array}{ll}
i=mp^{s-1}+r\text{,} & 0<r<p^{s-1}\text{,} \\
r=qp^{s-2}+r^{\prime }\text{,} & 0<r^{\prime }<p^{s-2}\text{.}%
\end{array}%
\end{equation*}%
So
\begin{equation*}
pi=(mp^{s-1}+r)p=pr=qp^{s-1}+pr^{\prime }\text{.}
\end{equation*}%
Without loss of generality assume that $r>pr^{\prime }$. Then,%
\begin{eqnarray*}
\phi_{L} (i)+\phi_{L} (pi) &=&\overline{1}_{r}\overline{0}_{p^{s-1}-r}+%
\overline{m}_{p^{s-1}}+\overline{q}_{p^{s-1}}+\overline{1}_{pr^{\prime }}%
\overline{0}_{p^{s-1}-pr^{\prime }} \\
&=&\overline{q+m+2}_{pr^{\prime }}\overline{q+m+1}_{r-pr^{\prime }}\overline{%
q+m}_{p^{s-1}-r}\notin \phi_{L} (C)\text{,}
\end{eqnarray*}%
since $0<pr^{\prime }<p^{s-1}$, $r-pr^{\prime }\neq 0$ and $p^{s-1}-r\neq 0$.

\item[(iii)] If $p^{s}-p^{s-1}<i<p^{s}$ then $0<-i<p^{s-1}$. So $\phi_{L}
(-i)+\phi_{L} (-pi)\notin \phi_{L} (C)$ as we proved in the first case. We
see that for each $v\in \mathbb{Z}_{p^{s}}^{n}$ we have either $\phi_{L}
(v)+\phi_{L} (pv)\notin \phi_{L} (C)$ or $\phi_{L} (-v)+\phi_{L} (-pv)\notin
\phi_{L} (C)$. Hence either $\phi_{L} (v)+\phi_{L} (C)\neq \phi_{L} (C)$ or $%
\phi_{L} (-v)+\phi_{L} (C)\neq \phi_{L} (C)$ when $ord(v)>p^{2}$.
\end{description}
\end{proof}

So the Gray image of the code, which is generated by all vectors of $C$ with
order less than or equal to $p^{2}$ should include $K(\phi_{L} (C))$. Then
we have the following corollary and lemmas, which generalize the results  in \cite{dougherty}:

\begin{corollary}
\label{kernelinccor}Let $C$ be a linear code over $\mathbb{Z}_{p^{s}}$\ with
generator matrix of the form (\ref{genmatc}).Then $K(\phi_{L} (C))$ is
contained in the Gray image of the code generated by the matrix:%
\begin{equation}
\left[
\begin{array}{ccccccc}
p^{s-2}I_{\delta _{0}} & p^{s-2}A_{0,1} & p^{s-2}A_{0,2} & p^{s-2}A_{0,3} &
\cdot \cdot \cdot & \cdot \cdot \cdot & p^{s-2}A_{0,s} \\
0 & p^{s-2}I_{\delta _{1}} & p^{s-2}A_{1,2} & p^{s-2}A_{1,3} & \cdot \cdot
\cdot & \cdot \cdot \cdot & p^{s-2}A_{1,s} \\
0 & 0 & p^{s-2}I_{\delta _{2}} & p^{s-2}A_{2,3} & \cdot \cdot \cdot & \cdot
\cdot \cdot & p^{s-2}A_{2,s} \\
\cdot \cdot \cdot & \cdot \cdot \cdot & 0 & \cdot \cdot \cdot & \cdot \cdot
\cdot & \cdot \cdot \cdot & \cdot \cdot \cdot \\
\cdot \cdot \cdot & \cdot \cdot \cdot & \cdot \cdot \cdot & \cdot \cdot \cdot
& \cdot \cdot \cdot & \cdot \cdot \cdot & \cdot \cdot \cdot \\
0 & 0 & 0 & \cdot \cdot \cdot & p^{s-2}I_{\delta _{s-2}} & p^{s-2}A_{s-2,s-1}
& p^{s-2}A_{s-2,s} \\
0 & 0 & 0 & \cdot \cdot \cdot & 0 & p^{s-1}I_{\delta _{s-1}} &
p^{s-1}A_{s-1,s}%
\end{array}%
\right] \text{.}  \label{genmatinc}
\end{equation}
\end{corollary}

\begin{lemma}
\label{sumlemma}Let $C$ be a linear code over $%
\mathbb{Z}
_{p^{s}}$ and $v,w\in C$. Then we have%
\begin{equation*}
\phi_{L} (p^{s-1}v+w)=\phi_{L} (p^{s-1}v)+\phi_{L} (w)
\end{equation*}%
for each $v,w\in C$.
\end{lemma}

\begin{proof}
Let $v_{i},w_{i}\in
\mathbb{Z}
_{p^{s}}$ be the $i^{th}$ coordinates of $v,w$ respectively. Then by
division algorithm we can write%
\begin{equation*}
\begin{array}{lll}
w_{i}=q_{w}p^{s-1}+r_{w}\text{,} & 0\leq q_{w}\leq p-1\text{,} & 0\leq
r_{w}<p^{s-1}\text{,} \\
v_{i}=q_{v}p^{s-1}+r_{v}\text{,} & 0\leq q_{v}<p^{s-1}\text{,} & 0\leq
r_{v}<p\text{.}%
\end{array}%
\end{equation*}%
So $p^{s-1}v=p^{s-1}r_{v}$, where $0\leq p^{s-1}r_{v}<p^{s}$. Therefore%
\begin{eqnarray*}
\phi_{L} (p^{s-1}v_{i}+w_{i}) &=&\phi_{L}
(p^{s-1}r_{v}+q_{w}p^{s-1}+r_{w})=\phi_{L} (p^{s-1}(r_{v}+q_{w})+r_{w}) \\
&=&\overline{r_{v}+q_{w}}_{p^{s-1}}+\overline{1}_{r_{w}}\overline{0}%
_{p^{s-1}-r_{w}}=\overline{r_{v}}_{p^{s-1}}+\overline{q_{w}}_{p^{s-1}}+%
\overline{1}_{r_{w}}\overline{0}_{p^{s-1}-r_{w}} \\
&=&\phi_{L} (p^{s-1}r_{v})+\phi_{L} (q_{w}p^{s-1}+r_{w})=\phi_{L}
(p^{s-1}v_{i})+\phi_{L} (w_{i})\text{.}
\end{eqnarray*}%
Applying this method coordinate-wise, the result follows.
\end{proof}

\begin{theorem}
\label{kernelsubcodethm}Let $C$ be a linear code over $%
\mathbb{Z}
_{p^{s}}$ with the generator matrix of the form (\ref{genmatc}). Then the
Gray image of the code $C'$ generated by
\begin{equation}
\left[
\begin{array}{ccccccc}
p^{s-1}I_{\delta _{0}} & p^{s-1}A_{0,1} & p^{s-1}A_{0,2} & p^{s-1}A_{0,3} &
\cdot \cdot \cdot & \cdot \cdot \cdot & p^{s-1}A_{0,s} \\
0 & p^{s-1}I_{\delta _{1}} & p^{s-1}A_{1,2} & p^{s-1}A_{1,3} & \cdot \cdot
\cdot & \cdot \cdot \cdot & p^{s-1}A_{1,s} \\
0 & 0 & p^{s-1}I_{\delta _{2}} & p^{s-1}A_{2,3} & \cdot \cdot \cdot & \cdot
\cdot \cdot & p^{s-1}A_{2,s} \\
\cdot \cdot \cdot & \cdot \cdot \cdot & 0 & \cdot \cdot \cdot & \cdot \cdot
\cdot & \cdot \cdot \cdot & \cdot \cdot \cdot \\
\cdot \cdot \cdot & \cdot \cdot \cdot & \cdot \cdot \cdot & \cdot \cdot \cdot
& \cdot \cdot \cdot & \cdot \cdot \cdot & \cdot \cdot \cdot \\
0 & 0 & 0 & \cdot \cdot \cdot & p^{s-1}I_{\delta _{s-2}} & p^{s-1}A_{s-2,s-1}
& p^{s-1}A_{s-2,s} \\
0 & 0 & 0 & \cdot \cdot \cdot & 0 & p^{s-1}I_{\delta _{s-1}} &
p^{s-1}A_{s-1,s}%
\end{array}%
\right]  \label{genmatcdouble}
\end{equation}%
is a linear subcode of $K(\phi_{L} (C))$.
\end{theorem}

\begin{proof}
Let $v,w\in C$, then $p^{s-1}v\in C'\subseteq C$. Then $\phi_{L}
(p^{s-1}v)\in \phi_{L} (C')$ and $\phi_{L} (w)\in \phi_{L} (C)$. By Lemma %
\ref{sumlemma}%
\begin{equation*}
\phi_{L} (p^{s-1}v+w)=\phi_{L} (p^{s-1}v)+\phi_{L} (w)\in \phi_{L} (C)\text{,%
}
\end{equation*}%
since $p^{s-1}v,w\in C$. This holds for every $w\in C$, which means $%
\phi_{L} (p^{s-1}v)+\phi_{L} (C)\subseteq \phi_{L} (C)$. Two different
codewords will have different images. Therefore $\phi_{L}
(p^{s-1}v)+\phi_{L} (C)=\phi_{L} (C)$, which tells us that $\phi_{L}
(p^{s-1}v)\in K(\phi_{L} (C))$.
\end{proof}

\begin{lemma}
\label{lambdavlem}Let $C$ be a linear code over $%
\mathbb{Z}
_{p^{s}}$, $\lambda \in
\mathbb{Z}
_{p^{s}}$ and $v\in C$ such that $\phi_{L} (v)\notin K(\phi_{L} (C))$. Then $%
\phi_{L} (\lambda v)\in K(\phi_{L} (C))$ if and only if $ord(\lambda v)=p$.
\end{lemma}

\begin{proof}
$(\Longrightarrow )$:Suppose that $ord(\lambda v)=p$, then $\phi_{L}
(\lambda v)\in K(\phi_{L} (C))$ by Theorem \ref{kernelsubcodethm}.

$(\Longleftarrow )$Now assume $\phi_{L} (v)\notin K(\phi_{L} (C))$ and $%
\phi_{L} (\lambda v)\in K(\phi_{L} (C))$. We have two cases.

\begin{description}
\item[(i)] If $ord(v)>p^{2}$ and $v=(v_{1},v_{2},\cdot \cdot \cdot ,v_{n})$,
then there exists $v_{i}$, $1\leq i\leq n$, such that $ord(v_{i})>p^{2}$.
Let $ord(v_{i})=p^{k}$ with $k>2$. Then $v_{i}=p^{s-k}u_{i}$, where $u_{i}$
is a unit. By division algorithm, we have%
\begin{equation*}
\begin{array}{lll}
u_{i}=q_{u}p+r_{u}\text{,} & 0\leq q_{u}\leq p^{s-1}-1\text{,} & 0<r_{u}<p%
\text{,} \\
v_{i}=q_{v}p^{s-1}+r_{v}\text{,} & 0\leq q_{v}\leq p-1\text{,} &
0<r_{v}<p^{s-1}\text{,}%
\end{array}%
\end{equation*}%
where $r_{u}\neq 0$, since $u_{i}$ is a unit and $r_{v}\neq 0$, since $%
ord(v_{i})>p^{2}$. If $\phi_{L} (\lambda v_{i})\in K(\phi_{L} (C))$, then by
Proposition \ref{psquareprop} $\lambda =p^{k-2}u_{\lambda }$ or $\lambda
=p^{k-1}u_{\lambda }$, where $u_{\lambda }$ is a unit. For $\lambda
=p^{k-1}u_{\lambda }$ we have $ord(\lambda v_{i})=p$, so $\phi_{L} (\lambda
v_{i})\in K(\phi_{L} (C))$ by Theorem \ref{kernelsubcodethm}. If $\lambda
=p^{k-2}u_{\lambda }$, then $ord(\lambda v_{i})=p^{2}$ and $\lambda
v_{i}=p^{s-2}u_{\lambda }u_{i}=q_{u}u_{\lambda }p^{s-1}+r_{u}u_{\lambda
}p^{s-2}$, where $0<r_{u}u_{\lambda }p^{s-2}<p^{s-1}$. Without loss of
generality assume that $r_{u}u_{\lambda }p^{s-2}<r_{v}$, then we have
\begin{equation*}
\phi_{L} (\lambda v_{i})+\phi_{L} (v_{i})=\overline{(q_{u}+q_{v}+2)}%
_{r_{u}u_{\lambda }p^{s-2}}\overline{(q_{u}+q_{v}+1)}_{r_{v}-r_{u}u_{\lambda
}p^{s-2}}\overline{(q_{u}+q_{v})}_{p^{s-1}-r_{v}}\notin \phi_{L} (C)\text{,}
\end{equation*}%
since $r_{u}u_{\lambda }p^{s-2}\neq 0$, $r_{v}-r_{u}u_{\lambda }p^{s-2}\neq
0 $, $p^{s-1}-r_{v}\neq 0$. If $r_{u}u_{\lambda }p^{s-2}=r_{v}$, then
\begin{equation*}
\phi_{L} (\lambda v_{i})+\phi_{L} (v_{i})=\overline{(q_{u}+q_{v}+2)}_{r_{v}}%
\overline{(q_{u}+q_{v})}_{p^{s-1}-r_{v}}\notin \phi_{L} (C)\text{,}
\end{equation*}%
since $p^{s-1}-r_{v}\neq 0$, $r_{v}\neq 0$.

\item[(ii)] If $ord(v)=p^{2}$ and $v=(v_{1},v_{2},\cdot \cdot \cdot ,v_{n})$%
, then there exists $v_{i}$, $1\leq i\leq n$, such that $ord(v_{i})=p^{2}$.
Then $v_{i}=p^{s-2}u_{i}$, where $u_{i}$ is a unit. By division algorithm,
we have $v_{i}=q_{v}p^{s-1}+r_{v}$, $0\leq q_{v}\leq p-1$, $0<r_{v}<p^{s-1}$%
, since $ord(v_{i})=p^{2}$. If $\phi_{L} (\lambda v_{i})\in K(\phi_{L} (C))$%
, then by Proposition \ref{psquareprop} $ord(\lambda v_{i})=p^{2}$ or $%
ord(\lambda v_{i})=p$. If $ord(\lambda v_{i})=p$, we have $\phi_{L} (\lambda
v_{i})\in K(\phi_{L} (C))$ by Theorem \ref{kernelsubcodethm}. If $%
ord(\lambda v_{i})=p^{2}$ then $\lambda $ is a unit and $\lambda
v_{i}=p^{s-1}q+r$, $0<r<p^{s-1}$, $r\neq 0$, since $ord(\lambda v_{i})=p^{2}$%
. Without loss of generality assume that $r_{v}>r$, then we have%
\begin{equation*}
\phi_{L} (\lambda v_{i})+\phi_{L} (v_{i})=\overline{(q+q_{v}+2)}_{r}%
\overline{(q+q_{v}+1)}_{r_{v}-r}\overline{(q_{u}+q_{v})}_{p^{s-1}-r_{v}}%
\notin \phi_{L} (C)\text{,}
\end{equation*}%
since $r\neq 0$, $r_{v}-r\neq 0$, $p^{s-1}-r_{v}\neq 0$. If $r=r_{v}$, then%
\begin{equation*}
\phi_{L} (\lambda v_{i})+\phi_{L} (v_{i})=\overline{(q_{u}+q_{v}+2)}_{r_{v}}%
\overline{(q_{u}+q_{v})}_{p^{s-1}-r_{v}}\notin \phi_{L} (C)\text{,}
\end{equation*}%
since $p^{s-1}-r_{v}\neq 0$, $r_{v}\neq 0$. In both cases $\phi_{L} (\lambda
v)+\phi_{L} (v)\notin \phi_{L} (C)$, whenever $ord(\lambda v)\neq p$.
\end{description}
\end{proof}

\begin{theorem}
\label{kerdimrestthm}Let $C$ be a linear code over $%
\mathbb{Z}
_{p^{s}}$ of type $\left\{ \delta _{0},\delta _{1},\cdot \cdot \cdot ,\delta
_{s-1}\right\} $. If $m=\dim (K(\phi_{L} (C)))$, then%
\begin{equation*}
m\in \left\{ \sum\limits_{i=0}^{s-1}\delta
_{i},\sum\limits_{i=0}^{s-1}\delta _{i}+1,\sum\limits_{i=0}^{s-1}\delta
_{i}+2,\cdot \cdot \cdot ,\sum\limits_{i=0}^{s-1}\delta _{i}+\delta
_{s-2}-2,\sum\limits_{i=0}^{s-1}\delta _{i}+\delta _{s-2}\right\} .
\end{equation*}
\end{theorem}

\begin{proof}
By Theorem \ref{kernelsubcodethm}, the image of any codeword of order $p$ is
in $K(\phi_{L} (C))$. If there is a codeword $v$ of order greater than $p^{2}
$, then $\phi_{L} (v)\notin K(\phi_{L} (C))$. Moreover, if $\phi_{L}
(v)\notin K(\phi_{L} (C))$, then $\phi_{L} (\lambda v)\in K(\phi_{L} (C))$
if and only if $ord(\lambda v)=p$ by Lemma \ref{lambdavlem}. Otherwise $%
\phi_{L} (\lambda v)+\phi_{L} (v)\notin \phi_{L} (C)$. So for $\phi_{L}
(v)\notin K(\phi_{L} (C))$ and $\phi_{L} (\lambda v)\in \phi_{L}
(C")\subseteq $ $K(\phi_{L} (C))$ we have $ord(\lambda v)=p$. This means we
have the Gray images of first $\sum\limits_{i=0}^{s-3}\delta _{i}$ vectors
of (\ref{genmatcdouble}) as generators of $K(\phi_{L} (C))$. Furthermore, we
can show that the contribution of the Gray images of first $%
\sum\limits_{i=0}^{s-3}\delta _{i}$ vectors of (\ref{genmatc}) to $%
K(\phi_{L} (C))$ is restricted to that. To see this, let $v$ be one these
vectors in (\ref{genmatc}). Then $ord(v)>p^{2}$ and $\phi_{L} (v)\notin
K(\phi_{L} (C))$ by Proposition \ref{psquareprop}. For any scalar product of
$v$, say $\lambda v$, then $\phi_{L} (\lambda v)\in K(\phi_{L} (C))$ if and
only if $ord(\lambda v)=p$ by Lemma \ref{lambdavlem}. If $ord(v)=p^{k}$, $k>2
$, $v=u_{v}p^{k}$, this happens only when $\lambda =p^{s-k-1}u_{\lambda }$,
where $u_{\lambda }$ and $u_{v}$ are units. Therefore $\lambda v=p^{s-1}u$,
where $u=u_{v}u_{\lambda } $ is a unit too. This shows that the only
contribution of the Gray image of $v$ to $K(\phi_{L} (C))$ is its scalar
products with the   $p^{s-1}u$ and their linear combinations. Also we know that
the Gray image of the last $\delta _{s-1}$ rows of (\ref{genmatcdouble}) are
generators of $K(\phi_{L} (C))$ by Theorem \ref{kernelsubcodethm}. We don't
know whether each of the Gray images of $\delta _{s-2}$\ remaining vectors
generate $K(\phi_{L} (C))$ certainly. But we know that if their Gray images
are not included in generators of $K(\phi_{L} (C))$, the Gray image of their
scalar products with $pu $, where $u$ is a unit, are all included in $%
K(\phi_{L} (C))$. Hence we can have at least the Gray image of the code
generated by (\ref{genmatcdouble}), and at most the Gray image of the code
generated by%
\begin{equation}
\left[
\begin{array}{ccccccc}
p^{s-1}I_{\delta _{0}} & p^{s-1}A_{0,1} & p^{s-1}A_{0,2} & p^{s-1}A_{0,3} &
\cdot \cdot \cdot & \cdot \cdot \cdot & p^{s-1}A_{0,s} \\
0 & p^{s-1}I_{\delta _{1}} & p^{s-1}A_{1,2} & p^{s-1}A_{1,3} & \cdot \cdot
\cdot & \cdot \cdot \cdot & p^{s-1}A_{1,s} \\
0 & 0 & p^{s-1}I_{\delta _{2}} & p^{s-1}A_{2,3} & \cdot \cdot \cdot & \cdot
\cdot \cdot & p^{s-1}A_{2,s} \\
\cdot \cdot \cdot & \cdot \cdot \cdot & 0 & \cdot \cdot \cdot & \cdot \cdot
\cdot & \cdot \cdot \cdot & \cdot \cdot \cdot \\
\cdot \cdot \cdot & \cdot \cdot \cdot & \cdot \cdot \cdot & \cdot \cdot \cdot
& \cdot \cdot \cdot & \cdot \cdot \cdot & \cdot \cdot \cdot \\
0 & 0 & 0 & \cdot \cdot \cdot & p^{s-2}I_{\delta _{s-2}} & p^{s-2}A_{s-2,s-1}
& p^{s-2}A_{s-2,s} \\
0 & 0 & 0 & \cdot \cdot \cdot & 0 & p^{s-1}I_{\delta _{s-1}} &
p^{s-1}A_{s-1,s}%
\end{array}%
\right]  \label{genmatmost}
\end{equation}%
as $K(\phi_{L} (C))$. Thus we have the following bound for $m$:%
\begin{equation*}
p^{\sum\limits_{i=0}^{s-1}\delta _{i}}\leq p^{m}\leq
p^{\sum\limits_{i=0,i\neq s-2}^{s-1}\delta _{i}}\cdot p^{2\delta _{s-2}}%
\text{,}
\end{equation*}%
which means
\begin{equation*}
\sum\limits_{i=0}^{s-1}\delta _{i}\leq m\leq \sum\limits_{i=0}^{s-1}\delta
_{i}+\delta _{s-2}\text{.}
\end{equation*}%
Let $\widetilde{C}$ be the code generated by matrix (\ref{genmatmost}).
Since $K(\phi_{L} (C))$ is at most $\phi_{L} (\widetilde{C})$, $K(\phi_{L}
(C))\subseteq \phi_{L} (\widetilde{C})$. So,%
\begin{equation*}
K(\phi_{L} (C))=\left\{ c\in \widetilde{C}:\phi_{L} (c)+\phi_{L}
(C)=\phi_{L} (C)\right\} \text{.}
\end{equation*}%
Let $\left\{ v_{0},v_{1},\cdot \cdot \cdot ,v_{k}\right\} $ be the set of
generators of $\phi_{L} (\widetilde{C})$, namely $\left\langle
v_{0},v_{1},\cdot \cdot \cdot ,v_{k}\right\rangle =\phi_{L} (\widetilde{C})$%
, which means $\dim (\phi_{L} (\widetilde{C}))=k+1$. Assume that $\dim
(K(\phi_{L} (C)))=k$, and without loss of generality let $K(\phi_{L}
(C))=\left\langle v_{1},\cdot \cdot \cdot ,v_{k}\right\rangle $. If $v_{0}\in
\phi_{L} (\widetilde{C})\subseteq \phi_{L} (C)$, then we have $%
v_{0}+v_{i}\in \phi_{L} (C)$ for all $i=1,\cdot \cdot \cdot ,k$, since $%
v_{i}\in K(\phi_{L} (C))$. But $v_{0}+v_{i}\in \phi_{L} (\widetilde{C})$ for
all $i=1,\cdot \cdot \cdot ,k$, that means $v_{0}\in K(\phi_{L} (\widetilde{C%
}))\subseteq K(K(\phi_{L} (C)))\subseteq K(\phi_{L} (C))$, which is a
contradiction. Hence $m\neq \sum\limits_{i=0}^{s-1}\delta _{i}+\delta
_{s-2}-1$. Therefore we have the following%
\begin{equation*}
m\in \left\{ \sum\limits_{i=0}^{s-1}\delta
_{i},\sum\limits_{i=0}^{s-1}\delta _{i}+1,\sum\limits_{i=0}^{s-1}\delta
_{i}+2,\cdot \cdot \cdot ,\sum\limits_{i=0}^{s-1}\delta _{i}+\delta
_{s-2}-2,\sum\limits_{i=0}^{s-1}\delta _{i}+\delta _{s-2}\right\} \text{.}
\end{equation*}
\end{proof}

\section{Linearity and Duality of $\protect\phi_{L} (C)$}

Self-dual codes are important since many of the best codes known are of this
type. Numerous researchers have investigated their Gray images to find (not
necessarily linear) codes with optimal or extremal parameters. Most of the
best codes are nonlinear and they can be viewed as Gray images of linear
codes. On the other hand, linearity makes things easier. Therefore it is
also very important to know when the image $\phi_{L} (C)$ is
nonlinear/linear. Also some researchers looked into when the images of
self-dual codes are also self-dual. The aim of this section is to present
some knowledge about these two topics for codes over $%
\mathbb{Z}
_{p^{s}}$.

\begin{theorem}
Let $C$ be a linear code with the generating matrix of the form given in (%
\ref{genmatc}). If $\delta _{i}>0$ for $0\leq i\leq s-3$ then $\phi_{L} (C)$
is not linear.
\end{theorem}

\begin{proof}
We have elements $v\in C$ such that $ord(v)>p$, so by Lemma \ref{lambdavlem}
they are not in $K(\phi_{L} (C))$. Hence the image is not linear.
\end{proof}

\begin{theorem}
Let $C$ be a linear code. If $p>2$ then the image of a free code is not
linear.
\end{theorem}

\begin{proof}
If $C$ is a free code, then it has a generating matrix of the form%
\begin{equation*}
G=\left[
\begin{array}{cc}
I_{\delta _{0}} & A%
\end{array}%
\right] \text{,}
\end{equation*}%
where $A$ is an $\delta _{0}\times (n-\delta _{0})$ matrix over $%
\mathbb{Z}
_{p^{s}}$. Let $v_{i}$ be the $i^{th}$ row of $G$. Since every row of $G$ is
a codeword, if $\phi_{L} (C)$ is linear then $-\phi_{L} (v_{1})$ must be
included in $\phi_{L} (C)$. But
\begin{equation*}
-\phi_{L} (v_{1})=(-\phi_{L} (1),-\phi_{L} (v_{1,2}),\cdot \cdot \cdot
,-\phi_{L} (v_{1,n}))\notin \phi_{L} (C)\text{,}
\end{equation*}%
because $-\phi_{L} (1)\neq -\phi_{L} (x)$ for any $x\in
\mathbb{Z}
_{p^{s}}$.
\end{proof}

The image of a self-dual code $C$ over $%
\mathbb{Z}
_{p^{s}}$\ under the Gray map only has the cardinality of a self-dual code
if $p=2$ and $s=2$, since a self-dual code should include exactly half of
the ambient space, which means $\frac{sn}{2}=\frac{p^{s-1}n}{2}$. This
implies $s=p^{s-1}$ and hence $p=s=2$. So for $p>2$ we know that none of the
self-dual codes has self-dual image. However a code might have a self-dual
image if it is not self-dual. First we need to seek for self-orthogonal
images.

\begin{theorem}
Any code $C$ over $%
\mathbb{Z}
_{p^{s}}$\ of type $\left( p^{s-1}\right) ^{\delta _{1}}\left(
p^{s-2}\right) ^{\delta _{2}}\cdot \cdot \cdot \left( p^{2}\right) ^{\delta
_{s-2}}\left( p\right) ^{\delta _{s-1}}$ has an image that is a
self-orthogonal code.
\end{theorem}

\begin{proof}
If $C$ is of type $\left( p^{s-1}\right) ^{\delta _{1}}\left( p^{s-2}\right)
^{\delta _{2}}\cdot \cdot \cdot \left( p^{2}\right) ^{\delta _{s-2}}\left(
p\right) ^{\delta _{s-1}}$, then it has a generating matrix of the form%
\begin{equation*}
G=\left[
\begin{array}{cccccc}
pI_{\delta _{1}} & pA_{1,2} & pA_{1,3} & \cdot \cdot \cdot & \cdot \cdot
\cdot & pA_{1,s} \\
0 & p^{2}I_{\delta _{2}} & p^{2}A_{2,3} & \cdot \cdot \cdot & \cdot \cdot
\cdot & p^{2}A_{2,s} \\
0 & 0 & \cdot \cdot \cdot & \cdot \cdot \cdot & \cdot \cdot \cdot & \cdot
\cdot \cdot \\
\cdot \cdot \cdot & \cdot \cdot \cdot & \cdot \cdot \cdot & \cdot \cdot \cdot
& \cdot \cdot \cdot & \cdot \cdot \cdot \\
0 & 0 & 0 & p^{s-2}I_{\delta _{s-2}} & p^{s-2}A_{s-2,s-1} & p^{s-2}A_{s-2,s}
\\
0 & 0 & 0 & 0 & p^{s-1}I_{\delta _{s-1}} & p^{s-1}A_{s-1,s}%
\end{array}%
\right] \text{.}
\end{equation*}%
Let $v=(v_{1},\cdot \cdot \cdot ,v_{n}),w=(w_{1},\cdot \cdot \cdot
,w_{n})\in C$ are rows of $G$ with order $p^{s-i_{1}}$ and $p^{s-i_{2}}$,
where $i_{1}\geq i_{2}\geq 1$. So each $v_{k}$ is in $\left\{
0,p^{i_{1}},2p^{i_{1}},\cdot \cdot \cdot ,p^{s}-p^{i_{1}}\right\} $ and each
$w_{k}$ is in $\left\{ 0,p^{i_{2}},2p^{i_{2}},\cdot \cdot \cdot
,p^{s}-p^{i_{2}}\right\} $, where $1\leq k\leq n$. For any element $m$ in $%
\mathbb{Z}
_{p^{s}}$ of order $p^{s-e}$ we have%
\begin{equation*}
\phi_{L} (m)=\overline{(q+1)}_{p^{e}t}\overline{(q)}_{(p^{s-1-e}-t)p^{e}}%
\text{,}
\end{equation*}%
where $m=p^{s-1}q+r$, $0\leq r=p^{e}t<p^{s-1}$, $0\leq q\leq p-1$. We will
consider $\left\langle \phi_{L} (v_{k}),\phi_{L} (w_{k})\right\rangle $
instead of $\left\langle \phi_{L} (v),\phi_{L} (w)\right\rangle $, since $%
\phi_{L} (v)=(\phi_{L} (v_{1}),\cdot \cdot \cdot ,\phi_{L} (v_{n}))$, $%
\phi_{L} (w)=(\phi_{L} (w_{1}),\cdot \cdot \cdot ,\phi_{L} (w_{n}))$, and
therefore $\left\langle \phi_{L} (v),\phi_{L} (w)\right\rangle
=\sum\limits_{i=1}^{n}\left\langle \phi_{L} (v_{i}),\phi_{L}
(w_{i})\right\rangle $. In both images the number of successively repeated
coordinates are divisible by a power of $p$ (at least by $p$). So in
coordinatewise product $\phi_{L} (v_{k})\cdot \phi_{L}
(w_{k})=(v_{k,1}w_{k,1},\cdot \cdot \cdot ,v_{k,p^{s-1}}w_{k,p^{s-1}})$ the
coordinates will be repeated at least $p$ times successively. So $\phi_{L}
(v_{k})\cdot \phi_{L} (w_{k})=(\overline{(a_{1})}_{p},\overline{(a_{2})}%
_{p},\cdot \cdot \cdot ,\overline{(a_{p^{s-2}})}_{p})$, where $a_{l}$ is the
$l^{th}$ repeating coordinate. Hence%
\begin{equation*}
\left\langle \phi_{L} (v_{k}),\phi_{L} (w_{k})\right\rangle
=\sum\limits_{i=1}^{p^{s-1}}\left( \phi_{L} (v_{k})\cdot \phi_{L}
(w_{k})\right) _{i}=\sum\limits_{j=1}^{p^{s-2}}pa_{j}=0\text{,}
\end{equation*}%
which means $\phi_{L} (C)\subseteq \left( \phi_{L} (C)\right) ^{\bot }$.
\end{proof}

\end{document}